\documentclass[lotsofwhite,charterfonts, letter]{patmorin}
\usepackage{amsmath, amsthm,amsfonts,graphicx,color,marvosym}
\usepackage[sort&compress,numbers]{natbib}
\usepackage{enumitem}

\usepackage[letterpaper, total={6in,9in}]{geometry}

\usepackage[sc, small]{titlesec} 

\usepackage{amsthm}

\newcommand{\comment}[1]{}

\newcommand{\seclabel}[1]{\label{sec:#1}}

\newcommand{\secref}[1]{\mbox{Section~\ref{sec:#1}}}

\renewcommand{\eqref}[1]{(\ref{eq:#1})}

\newtheorem{thm}{Theorem}{\bfseries}{\itshape}
\newcommand{\thmlabel}[1]{\label{thm:#1}}
\newcommand{\thmref}[1]{Theorem~\ref{thm:#1}}

\newtheorem{lem}{Lemma}{\bfseries}{\itshape}
\newcommand{\lemlabel}[1]{\label{lem:#1}}
\newcommand{\lemref}[1]{Lemma~\ref{lem:#1}}

{\bfseries}{\itshape}

{\bfseries}{\itshape}

{\bfseries}{\itshape}

{\bfseries}{\itshape}

\newtheorem{assumption}{Assumption}{\bfseries}{\rm}

\newcommand{\etal}{\emph{et al.}}

\newcommand{\xx}{\ensuremath{\protect{x}}}
\newcommand{\yy}{\ensuremath{\protect{y}}}
\newcommand{\zz}{\ensuremath{\protect{z}}}

\newcommand{\fg}{\ensuremath{\mathcal{F}}}

\newcommand{\po}{\ensuremath{<_\mathcal{F}}}

\newcommand{\x}{\ensuremath{\protect\textup{\textsf{x}}}}
\newcommand{\y}{\ensuremath{\protect\textup{\textsf{y}}}}

\renewcommand{\thefootnote}{\fnsymbol{footnote}}

\title{\MakeUppercase{The Utility of Untangling}} 

\author{Vida Dujmovi\'c\footnote{School of Computer Science and Electrical Engineering, University of Ottawa, 
Ottawa, Canada, \texttt{vida.dujmovic@uottawa.ca}. Research supported by NSERC and the Ontario Ministry of Research and Innovation.}}

\date{}
\begin{document}

\maketitle

\renewcommand{\thefootnote}{\arabic{footnote}}



\begin{abstract}
In this note we show how techniques developed for untangling planar graphs by  Bose \etal\ [Discrete \& Computational Geometry 42(4): 570-585 (2009)] and Goaoc \etal\  [Discrete \& Computational Geometry 42(4): 542-569 (2009)]  imply new results about some recent graph drawing models. These include column planarity,  universal point subsets, and partial simultaneous geometric embeddings (with or without mappings). Some of these results answer open problems posed in previous papers. 
%
%
%
%
\end{abstract}
%
%

\section{Introduction}\seclabel{intro}
A \emph{geometric graph} is a graph whose vertex set is a set of distinct points in the plane and each pair of adjacent vertices $\{v, w\}$ is connected by a line segment $\overline{vw}$ that intersects only the two vertices. A geometric graph is \emph{planar} if its underlying combinatorial graph is planar. It is \emph{plane} if no two edges cross other than in a common endpoint. A \emph{straight-line crossing-free drawing}  of a planar graph is a representation of that graph by a plane geometric graph. 


Given a geometric planar graph, possibly with many crossings, to \emph{untangle} it, means to move some of its vertices to new locations  (that is, change their coordinates) such that the resulting geometric graph is plane. 
The goal is to do so by moving as few vertices as possible, or in other words, by keeping the locations of as many vertices as possible unchanged (that is, \emph{fixed}). A series of papers have studied untangling of planar graphs or subclasses of planar graphs \cite{DBLP:journals/dcg/PachT02, DBLP:journals/siamdm/CanoTU14, DBLP:journals/dcg/Cibulka10, DBLP:journals/dcg/BoseDHLMW09, DBLP:journals/dcg/GoaocKOSSW09, DBLP:journals/dam/KangPRSV11, DBLP:conf/wg/RavskyV11}. 
 The best known (lower) bound for general planar graphs is due to Bose \etal~\cite{DBLP:journals/dcg/BoseDHLMW09} who proved that every $n$-vertex geometric planar graph can be untangled while keeping the locations of at least $\Omega(n^{1/4})$ vertices fixed.  On the other hand, Cano~\etal~\cite{DBLP:journals/siamdm/CanoTU14}  showed that for all large enough $n$, there exists an $n$-vertex geometric planar graph that cannot be untangled while keeping the locations of more than $\omega(n^{0.4948})$ vertices fixed. 

The purpose of this note is to highlight how the techniques developed by Bose \etal\ \cite{DBLP:journals/dcg/BoseDHLMW09} and Goaoc \etal~\cite{DBLP:journals/dcg/GoaocKOSSW09}  can be used to establish new results on several recently studied graph drawing problems. 
Before presenting the new results we state the two key lemmas that are at the basis of all the results.  The statements of these two lemmas are new, but their proofs are  contained in and directly inferred by the work described in 
 \cite{DBLP:journals/dcg/BoseDHLMW09} and \cite{DBLP:journals/dcg/GoaocKOSSW09}.



Let $G$ be a plane triangulation (that is, an embedded simple planar graph each of whose faces is bounded by a $3$-cycle). 
Canonical orderings of plane triangulations were introduced by
de Fraysseix~\etal~\cite{dFPP90}. They proved that
$G$ has a vertex ordering $\sigma=(v_1:=\xx,v_2:=\yy, v_3, \dots,
v_n:=\zz)$, called a \emph{canonical ordering}, with the following
properties. Define $G_i$ to be the embedded subgraph of $G$ induced by
$\{v_1, v_2,\dots,v_i\}$.  Let $C_i$ be the subgraph of $G$ induced
by the edges on the boundary of the outer face of $G_i$. Then
\begin{itemize}[noitemsep, nolistsep]
\item \xx, \yy\ and \zz\ are the vertices on the outer face of $G$. 
\item  For each $i\in\{3,4,\dots,n\}$, $C_i$ is a cycle containing $\xx\yy$.
\item  For each $i\in\{3,4,\dots,n\}$, $G_i$ is biconnected and \emph{internally $3$-connected}; that is, removing any two interior vertices of $G_i$ does not disconnect it.
\item For each $i\in\{3,4,\dots,n\}$, $v_i$ is a vertex of $C_i$ with at least two neighbours in $C_{i-1}$, and these neighbours are consecutive on $C_{i-1}$.
\end{itemize}

The following structure was defined first in Bose~\etal~\cite{DBLP:journals/dcg/BoseDHLMW09}. Using the above notation, a \emph{frame} \fg\ of $G$ is the oriented subgraph of $G$ with vertex set  $V(\fg):=V(G)$, where: 
\begin{itemize}[noitemsep, nolistsep]
\item Edges \xx\yy, $\xx v_1$ and $v_1\yy$\ are in $E(\fg)$ where \xx\yy\ is oriented from \xx\ to \yy,  $\xx v_1$  is oriented from \xx\ to  $v_1$ and   $v_1\yy$\ is oriented from $v_1$ to $\yy$.
\item For each $i\in\{4,5\dots,n\}$ in the canonical ordering $\sigma$ of $G$, edges $pv_i$ and $v_ip'$ are in $E(\fg)$, where $p$ and $p'$ are the first and the last neighbour, respectively, of $v_i$ along the path in $C_{i-1}$ from \xx\ to \yy\ not containing edge \xx\yy. Edge $pv_i$ is oriented from $p$ to $v_i$, and edge $v_ip'$ is oriented from $v_i$ to $p'$. 
\end{itemize}

By definition, \fg\ is a directed acyclic graph with one source \xx\, and one sink \yy. \fg\ defines a partial order \po\ on $V(\fg)$, where $v\po w$ whenever there is a directed path from $v$ to $w$ in \fg. 

Subsequently, it has been observed that a frame of $G$ can also be obtained by taking the union of any two trees in  Schnyder 3-tree-decompositions where the  orientation of the edges in one of the two trees is reversed. See, for example, page $13$ in Di Giacomo~\etal~\cite{DBLP:journals/corr/abs-1212-0804} for this alternative formulation. 

Recall that a \emph{chain} (\emph{antichain}) in a partial order is a subset of its elements that are pairwise comparable (incomparable).
Given a partial order $(V,\leq)$ on a set of vertices $V$ of some graph, we will often refer to a \emph{chain $V'\subseteq V$}  (or antichain) and by that mean a subset of vertices of $V$ that form a chain (antichain) in the given partial order $(V, \leq)$.   We also say that a chain $V'$ contains a chain $V''$ if $V'$ and $V''$ are both chains in $(V,\leq)$ and $V''\subseteq V'$. 

Consider an $n$-vertex planar graph $G$ and a set $P$ of $k\leq n$ points in the plane together with a bijective mapping from a set $V_k$ of $k$ vertices in $G$ to $P$.
 Let $D$ be a straight-line crossing-free drawing of $G$. We say that $D$ \emph{respects} the given mapping if each vertex of $V_k$ is represented in $D$ by its image point as determined by the given mapping. 

The following two lemmas are implicit in the work of  Bose \etal\ \cite{DBLP:journals/dcg/BoseDHLMW09} and Goaoc \etal~\cite{DBLP:journals/dcg/GoaocKOSSW09}. Parts (b), (c) and consequently (d), in \lemref{chain}, are due to Goaoc \etal~\cite{DBLP:journals/dcg/GoaocKOSSW09}. Note that, unlike here, the results of Goaoc \etal~\cite{DBLP:journals/dcg/GoaocKOSSW09} are not expressed in terms of a chain in the frame of $G$ but an equivalent structure: a simple path $L$ in a plane triangulation, connecting two vertices $x$ and $y$ on the outer face $x,y,z$ with the property that all chords of $L$ lie on one side of $L$ and $z$ lies on the other. 

Consider a graph $G$, a set $S\subseteq V(G)$ and a set $P$ of $|S|$ points in the plane together with a bijective mapping from $S$ to $P$. For a vertex $v\in S$ mapped to a point $p\in P$, let $\x(v)$ denote the \x-coordinate of $p$. 

\begin{lem}\lemlabel{chain}~\cite{DBLP:journals/dcg/GoaocKOSSW09, DBLP:journals/dcg/BoseDHLMW09}
Let $G$ be an $n$-vertex plane triangulation with a partial order \po\ associated with a frame \fg\ of $G$. Let $C\subseteq V$ be a chain 
 in \po. Let $H$ be the graph induced in $G$ by a maximal chain that contains $C$ in \po. The embedding of $H$ is implied by the embedding of $G$. Then: 
\begin{enumerate}[noitemsep, nolistsep]
\item[(a)] $H$ is a $2$-connected outerplane graph, i.e. a $2$-connected embedded outerplanar graph all of whose vertices lie on the cycle bounding the infinite face.  
\item[(b)] Let $I\subseteq V(H)$ such that if $v, w\in I$ and $vw \in E(H)$ then $vw$ lies on the outer face of $H$.   Let $P$ be any set of $|I|$ points in the plane where no two points of $P$ have the same \x-coordinate. Given a bijective mapping from $I$ to $P$ such that, for every two vertices $v,w\in I$,  $v\po w$ if and only if $\x(v)<\x(w)$,
 there exists a straight-line crossing-free drawing of $G$ that respects the given mapping. 
\item[(c)] There exists such a set $I$ with at least $(V(H)+1)/2$ vertices.
\item[(d)] There exists such a set $I$ with at least $|C|/3$ vertices of $C$. 
\end{enumerate}
\end{lem}

While the lower bound in part (c)  is stronger than the lower bound in part (d), part (d) ensures that a fraction of vertices of $C$ are used. That will be critical for some applications (see \thmref{scp} in \secref{cp} and \thmref{ppsge} in \secref{psgwm}). Part (d) follows from (b) as follows. Consider the graph $H'$ induced in $H$ by the vertices of $C$. By part (a), $H'$ is outerplanar. Thus its vertices can be coloured with three colours such that adjacent vertices in $H'$ receive distinct colours. Thus there exists an independent set $I$ in $H'$ that contains at least $|C|/3$ vertices of $C$. The condition imposed on the vertex set $I$ in part (b) are immediate since $I$ is an independent set in $H'$ and $H$. 



Note that, in an interesting recent development, Di Giacomo~\etal~\cite{DBLP:conf/wg/GiacomoLM13} proved that every $n$-vertex plane 
triangulation has a frame where some chain has size at least $n^{1/3}$. Thus by part (a), $|V(H)|\geq n^{1/3}$ in that frame and consequently,  every $n$-vertex plane triangulation has a $2$-connected outerplane graph of size at least  $n^{1/3}$ as an embedded induced subgraph. 

The following is the second key lemma.

\begin{lem}\lemlabel{antichain}~\cite{DBLP:journals/dcg/BoseDHLMW09}
Let $G$ be an $n$-vertex plane triangulation with a partial order \po\ associated with a frame \fg\ of $G$ and the total order $<_\sigma$ associated with the corresponding canonical ordering. Let $A\subseteq V$ be an antichain in \po. Let $P$ be any set of $|A|$ points in the plane where no two points of $P$ have the same \x-coordinate.  Given a bijective mapping from $A$ to $P$ such that, for every two vertices $v,w\in A$,  $v<_\sigma w$ if and only if $\x(v)<\x(w)$,
 then there exists a straight-line crossing-free drawing of $G$ that respects the given mapping. 
\end{lem}

\section{Column Planarity}\seclabel{cp}
Given a planar graph $G$, a set $R\subset V(G)$ is \emph{column planar} in $G$ if the vertices of $R$ can be assigned \x-coordinates such that given any arbitrary assignment of \y-coordinates to $R$, there exists a straight-line crossing free drawing of $G$ that respects the implied mapping of vertices of $R$ to the plane. 

The column planar sets were first defined 
 by Evans~\etal~\cite{DBLP:conf/gd/EvansKSS14}. A slightly stronger notion\footnote{with the roles of \x\ and \y\ coordinates reversed} was used earlier (although not named) in \cite{DBLP:journals/dcg/BoseDHLMW09} (see Lemma 1 and Lemma 6 in \cite{DBLP:journals/dcg/BoseDHLMW09}) where such sets were studied and used to prove \lemref{antichain} in the previous section.  In particular, define a set $R\subset V(G)$ as \emph{strongly column planar}  if the following holds:  there exists a total order $\mu$ on $R$ such that 
\begin{enumerate}[noitemsep, nolistsep]
\item[(a)] given any set $P$ of $|R|$ points in the plane where no two points have the same \x-coordinate; and,
\item[(b)] given a bijective mapping from $R$ to $P$ such that, for every two vertices $v,w\in R$,  $v<_\mu w$ if and only if $\x(v)<\x(w)$,
\end{enumerate}
then there exists a straight-line crossing-free drawing of $G$ that respects the given mapping. 
%
%
Being strongly column planar implies being column planar but not the converse. We use this slightly stronger notion as it is needed in the later sections. 

Notions similar to column planarity were studied by Estrella-Balderrama \etal~\cite{DBLP:journals/comgeo/EstrellaBalderramaFK09} and 
Di Giacomo~\etal~\cite{DBLP:journals/jgaa/GiacomoDKLS09}.

It is implicit in the work of Bose~\etal~\cite{DBLP:journals/dcg/BoseDHLMW09}  (see the proof of Lemma 2 in \cite{DBLP:journals/dcg/BoseDHLMW09}) that every tree has a strongly column planar set of size at least $n/2$. For column planar sets, this result is improved to $14n/17$  by Evans~\etal~\cite{DBLP:conf/gd/EvansKSS14}. Having a bound greater than $n/2$ is critical for an application of column planarity to partial simultaneous geometric embedding with mapping \cite{DBLP:conf/gd/EvansKSS14}.  Barba~\etal~\cite{BHK} prove that every $n$-vertex outerplanar graph has a column planar set of size at least $n/2$.\footnote{We suspect that the results and proofs in both \cite{BHK}  and \cite{DBLP:conf/gd/EvansKSS14} also hold for strongly column planar sets but we have not verified that.}

Evans~\etal~\cite{DBLP:conf/gd/EvansKSS14} pose as an open problem the question of developing any bound for column planar sets in general planar graphs. We provide here the first non-trivial (that is, better than constant) bound for this problem.

\begin{thm}\thmlabel{cp}
For every $n$, every $n$-vertex planar graph $G$ has a (strongly) column planar set of size at least $\sqrt{n/2}$.
\end{thm}

\begin{proof}
If $|V(G)|\leq 2$, the result is trivially true. Thus we may assume that $G$ is a triangulated plane graph. Let \fg\ be a frame of $G$, let \po\ be its associated partial order, and let $\sigma$ be the associated canonical ordering. Consider a chain in \po\ of maximum size. (Hence, the chain starts with $x$ and ends with $y$). Let $H$ be the subgraph of $G$ induced by that chain, as defined in \lemref{chain}. Let $I\subseteq V(H)$ be as defined in \lemref{chain} (b). Consider any set $P$ of $|I|$ points in the plane where no two points have the same \x-coordinate and consider a bijective mapping from $I$ to $P$ such that, for every two vertices $v,w\in I$, it holds that $v\po w$ if and only if $\x(v)<\x(w)$.
By \lemref{chain}  (b),  there exists a straight-line crossing-free drawing of $G$ that respects the given mapping and thus $I$, as ordered by \po, is a strongly column planar set.  By \lemref{chain}  (c), $|I|\geq |V(H)|/2$.  Thus if the size of the maximum chain in \po\ is at least $\sqrt{2n}$, and thus $|V(H)|\geq \sqrt{2n}$,  we are done. Otherwise, by Dilworth's theorem \cite{Dilworth50}, \po\ has a partition into at most $\sqrt{2n}$ antichains. By the pigeon-hole principle, there is an antichain in that partition with at least $n/\sqrt{2n}=\sqrt{n/2}$ vertices. Let $A\subseteq V(G)$ be the maximum antichain in \po.  Consider any set $P$ of $|A|$ points in the plane where no two points have the same \x-coordinate and consider a bijective mapping from $A$ to $P$ such that, for every two vertices $v,w\in A$, it holds that  $v<_\sigma w$ if and only if $\x(v)<\x(w)$.
By \lemref{antichain},  there exists a straight-line crossing-free drawing of $G$ that respects the given mapping and thus $A$, as ordered by $<_\sigma$, is a strongly column planar set. This completes the proof since $|A|\geq \sqrt{n/2}$. \qed
\end{proof}

We conclude this section by proving a slightly stronger statement (with a slightly weaker bound when $S=V$) than \thmref{cp}. This stronger statement relies on part (d) of \lemref{chain}, and is a critical strengthening for some applications, such as partial simultaneous geometric embeddings with mappings (see \thmref{ppsge} in \secref{psgwm}).

\begin{thm}\thmlabel{scp}
Given any planar graph $G$  and any subset $S\subseteq V$, there exists $R\subseteq S$ such that $R$ is a strongly column planar set of $G$ and  $|R|\geq \sqrt{|S|/3}$.
\end{thm}

\begin{proof} 
If $|V(G)|\leq 2$, the result is trivially true. Thus we may assume that $G$ is a triangulated plane graph. Let \fg\ be a frame of $G$, let \po\ be its associated partial order, and let  $\sigma$ be the associated canonical ordering.
Assume first that \po\ has a chain $C$ such that $C\subseteq S$ and  $|C|\geq \sqrt{3|S|}$. Let $H$ be the subgraph of $G$ induced by a maximal chain that contains $C$ in \po, as defined in \lemref{chain}. Let $I\subseteq S$ be as defined in \lemref{chain}, (b) and (d). Consider any set $P$ of $|I|$ points in the plane where no two points have the same \x-coordinate and consider a bijective mapping from $I$ to $P$ such that, for every two vertices $v,w\in I$,  it holds that $v\po w$ if and only if $\x(v)<\x(w)$.
By \lemref{chain} (b),  there exists a straight-line crossing-free drawing of $G$ that respects the given mapping and thus $I$, as ordered by \po, is a strongly column planar set.  By \lemref{chain} (d), $I\subseteq C$ and $|I|\geq |C|/3|$.  Thus if  \po\ has a chain $C$ such that $C\subseteq S$ and  $|C|\geq \sqrt{3|S|}$, we are done. Otherwise, by Dilworth's theorem \cite{Dilworth50}, \po, when restricted to $S$, has a partition into at most $\sqrt{3|S|}$ antichains. By the pigeon-hole principle, there is an antichain $A\subseteq S$ in that partition that has at least $|S|/\sqrt{3|S|}=\sqrt{|S|/3}$ elements. Consider any set $P$ of $|A|$ points in the plane where no two points have the same \x-coordinate and consider a bijective mapping from $A$ to $P$, such that for every two vertices $v,w\in A$,  $v<_\sigma w$ if and only if $\x(v)<\x(w)$.
By \lemref{antichain},  there exists a straight-line crossing-free drawing of $G$ that respects the given mapping and thus $A$, as ordered by $<_\sigma$, is a strongly column planar set. This completes the proof since $|A|\geq \sqrt{|S|/3}$ and $A\subseteq S$. \qed
\end{proof}

\section{Universal Point Subsets}\seclabel{ups}
A set of points $P$ is \emph{universal} for a set of planar graphs if every graph from the set has a straight-line crossing-free drawing where each of its vertices maps to a distinct point in $P$.  
It is known that,  for all large enough $n$,  universal pointsets of size $n$ do not exist for all $n$-vertex planar graphs -- as first proved by de Fraysseix~\etal~\cite{dFPP90}. The authors also proved that the $O(n) \times O(n) $ integer grid is universal for all $n$-vertex planar graphs and thus a universal pointsets of size $O(n^2)$ exists. Currently the best known lower bound  on the size of a smallest universal pointset for $n$-vertex planar graphs
is $1.235n-o(n)$ \cite{DBLP:journals/ipl/Kurowski04} and the best known upper bound is $n^2/4 - O(n)$ \cite{DBLP:journals/jgaa/BannisterCDE14}. Closing the gap between $\Omega(n)$ and $O(n^2)$ is a major, and likely difficult, graph drawing problem, open since $1988$ \cite{deFraysseix:1988:SSS:62212.62254, dFPP90}.

This motivated the following notion introduced by Angelini~\etal~\cite{DBLP:conf/isaac/AngeliniBEHLMMO12}.  A set $P$ of $k\leq n$ points in the plane is a \emph{universal point subset} for all $n$-vertex planar graphs if the following holds: every $n$-vertex planar graph $G$ has a subset $S\subseteq V(G)$ of $k$ vertices and a bijective mapping from $S$ to $P$ such that there exists a straight-line crossing-free drawing of $G$ that respects that mapping. 
%

Angelini~\etal~\cite{DBLP:conf/isaac/AngeliniBEHLMMO12} proved that for every $n$ there exists a set of points of size at least $\sqrt{n}$ that is a universal point subset for all $n$-vertex planar graphs. 
Di~Giacomo~\etal~\cite{DBLP:journals/corr/abs-1212-0804}  continued this study and showed that for every $n$,  \emph{every} set $P$ of at most $(\sqrt{\log_2n}-1)/4$ points in the plane is a universal point subset for \emph{all} $n$-vertex planar graphs.  They also showed that \emph{every} one-sided convex point set $P$ of at most $n^{1/3}$ points in the plane is a universal point subset for \emph{all} $n$-vertex planar graphs. The following theorem improves all these results.  


\begin{thm}\thmlabel{ups}
Every set $P$ of at most $\sqrt{n/2}$ points in the plane is a universal point subset for all $n$-vertex planar graphs. 
\end{thm}

The proof of this lemma can be derived directly from \lemref{chain} and \lemref{antichain}, similarly to the proof of \thmref{cp}, but we will instead prove it using \thmref{cp}.

\begin{proof}
Rotate $P$ to obtain a new pointset $P'$ where no two points of $P'$ have the same \x-coordinate. By \thmref{cp}, every $n$-vertex planar graph has a strongly column planar set $R$ of size $|P|$. Thus, by the definition of strongly column planar sets, there exists a total order $\mu$ on $R$ such that given a bijective mapping from $R$ to $P'$ where for every two vertices $v,w\in R$,  $v<_\mu w$ if and only if $\x(v)<\x(w)$,
 there exists a straight-line crossing-free drawing of $G$ that respects the given mapping. Such a mapping clearly exists since  no two points of $P'$ have the same \x-coordinate.
Rotating $P'$ back to the original pointset completes the proof.   \qed
%
%
\end{proof}

It is not known if, for all $n$, there exist a universal point subset of size $n^{1/2+\epsilon}$ for some $\epsilon>0$. Better bounds are only known for outerplanar graphs. Namely, every pointset of size $n$ in general position is universal for all $n$-vertex  outerplanar graphs \cite{GMPP, DBLP:journals/comgeo/Bose02, DBLP:conf/cccg/CastanedaU96}. Should the results of Barba~\etal~\cite{BHK} apply to strongly column planar sets, then arguments equivalent to those above would show that every pointset of size $n/2$ is a universal point subset for all $n$-vertex outerplanar graphs.

\section{(Partial) Simultaneous Geometric Embeddings} 
\seclabel{psg}

\emph{Simultaneous Geometric Embeddings} were introduced by Bra{\ss}~\etal~\cite{DBLP:journals/comgeo/BrassCDEEIKLM07}.  Initially there were two main variants of this problem, one in which the mapping between the vertices of the two graphs is given and another in which the mapping is not given. Since then there has been a plethora of work on the subject for various variants of the problem -- see, for example a survey by Bl\"asius~\etal~\cite{BKR-HGD}.

\subsection{Without mapping}\seclabel{psgwom}
Whether the following statement, on simultaneous geometric embeddings, is true is an open question asked by  Bra{\ss}~\etal~\cite{DBLP:journals/comgeo/BrassCDEEIKLM07} in 2003: For all $n$ and for any two $n$-vertex planar graphs there exists a pointset $P$ of size $n$ 
 such that each of the two graphs has a straight-line crossing-free drawing with its vertices mapped to distinct points of $P$. The statement is known \emph{not} to be true when ``two'' is replaced by $7393$ and $n=35$ \cite{DBLP:conf/tjjccgg/CardinalHK12}.

This motivates a study of (partial) geometric  simultaneous embeddings -- various versions of which  have been proposed and studied in the literature \cite{BKR-HGD}. We start with the following version.

Two graphs $G_1$ and $G_2$, where $|V(G_1)|\geq |V(G_2)|$ 
are said to have a \emph{geometric  simultaneous embedding with no mapping} if there exists a pointset $P$ of size $|V(G_1)|$ such that each of the two graphs has a straight-line  crossing-free drawing where all of its vertices are mapped to distinct points in $P$. Angelini~\etal~\cite{jgt-AEFG-2015} write: ``What is the largest $k \leq n$ such that every $n$-vertex planar graph and every $k$-vertex planar graph admit a geometric simultaneous embedding with no mapping? Surprisingly, we are not aware of any super-constant lower bound for the value of $k$.''

The following theorem answers their questions. 

\begin{thm}\thmlabel{fab}
For every $n$ and every $k\leq \sqrt{n/2}$, every $n$-vertex planar graph and every $k$-vertex planar graph admit a geometric simultaneous embedding with no mapping.
\end{thm}

\begin{proof}
Let $G_1$ and $G_2$ be the two given planar graphs with $|V(G_1)|=n$ and $|V(G_2)|=k$.  By F\'ary's theorem, $G_2$ has a straight-line crossing-free drawing on some set, $P_2$, of $k$ points. By \thmref{ups}, $G_1$ has a straight-line crossing-free drawing where $|P_2|$  vertices of $G_1$ are mapped to distinct points in $P_2$. Consider now the set of points, $P$, defined by the vertices in the drawing of $G_1$. This set is our desired pointset as it is a set of $n$ points such that each of $G_1$ and $G_2$ has a straight-line crossing-free  drawing where all of its vertices are mapped to the points in $P$. \qed
\end{proof}

Here is another variant of the (partial) geometric  simultaneous embedding problem. For $k\leq n$, two $n$-vertex planar graphs $G_1$ and $G_2$ are said to have a  \emph{$k$-partial simultaneous geometric embedding with no mapping} ($k$-PSGENM) if there exists a set $P$ of at least $k$ points in the plane such that  each of the two graphs has a straight-line crossing-free drawing where $|P|$ of its vertices are mapped to distinct points of $P$. 
 Recall that Angelini~\etal~\cite{DBLP:conf/isaac/AngeliniBEHLMMO12} proved that for every $n$ there exists a set of points of size at least $\sqrt{n}$ that is a universal point subset for all $n$-vertex planar graphs. This implies that, for all $n$, any two $n$-vertex planar graphs have an $\sqrt{n}$-partial simultaneous geometric embedding with no mapping. Note however that this does not imply \thmref{fab}. Namely, if one starts with a straight-line crossing-free drawing of the smaller   graph $G_2$ (say on $\sqrt{n}$ vertices), there is no guarantee with this result that the bigger, $n$-vertex graph, $G_1$ can be drawn while using all the points generated by the drawing of $G_2$.

\subsection{With Mapping}\seclabel{psgwm}

The notion of \emph{$k$-partial simultaneous geometric embedding with mapping} ($k$-PSGE) is the same as $k$-PSGENM  except that a bijective mapping between $V(G_1)$ and $V(G_2)$ is given and the two drawings have a further restriction that if $v\in V(G_1)$ is mapped to a point in $P$ then the vertex $w$ in $V(G_2)$ that $v$ maps to, has to be mapped to the same point in $P$.  In other words, two $n$-vertex planar graphs $G_1$ and $G_2$ on \emph{the same vertex set}, $V$, are said to have a \emph{$k$-partial simultaneous geometric embedding with mapping} ($k$-PSGE) if there exists a straight-line crossing free drawing $D_1$ of $G_1$ and  $D_2$ of $G_2$ such that there exists a subset $V'\subseteq V$ with $|V'|\geq k$ and each vertex $v\in V'$ is represented by the same point in $D_1$ and $D_2$.

It is known that, for every large enough $n$, there are pairs of $n$-vertex planar graphs that do not have an $n$-partial simultaneous geometric embedding with mapping, that is, an  $n$-PSGE \cite{DBLP:journals/comgeo/BrassCDEEIKLM07}. In fact the same is true for  simpler families of planar graphs, for example for a tree and a path \cite{DBLP:journals/jgaa/AngeliniGKN12}, for a planar graph and a matching \cite{DBLP:journals/jgaa/AngeliniGKN12} and for three paths \cite{DBLP:journals/comgeo/BrassCDEEIKLM07}.

%

$k$-PSGE was introduced by Evans~\etal~\cite{DBLP:conf/gd/EvansKSS14} who proved (using their column planarity result) that any two $n$-vertex trees have an $11n/17$-PSGE.  
Barba~\etal~\cite{BHK} proved that any two $n$-vertex outerplanar graphs have an  $n/4$-PGSE.  Evans~\etal~\cite{DBLP:conf/gd/EvansKSS14}  also observed  that the main untangling result by Bose~\etal\ \cite{DBLP:journals/dcg/BoseDHLMW09} implies that every pair of $n$-vertex planar graphs has an $\Omega(n^{1/4})$-PSGE.  Namely, start with a straight-line crossing-free drawing of $G_1$. Since the vertex sets of $G_1$ and $G_2$ are the same, the drawing of $G_1$ (or rather the drawing of its vertex set) defines a straight-line drawing of $G_2$. Untangling $G_2$ such that $\Omega(n^{1/4})$ of its vertices remain fixed (which is possible by \cite{DBLP:journals/dcg/BoseDHLMW09}) gives the result.  

\begin{thm}\thmlabel{psge}\cite{BHK}
Every pair of $n$-vertex planar graphs has an $\Omega(n^{1/4})$-partial simultaneous geometric embedding with mapping, that is, it has an $\Omega(n^{1/4})$-PGSE. 
\end{thm}

However, the above untangling argument fails if we try to apply it one more time. Namely, consider the following generalization of the $k$-PGSE problem. Given any set  $\{G_1, \dots, G_p\}$ of $p\geq 2$  $n$-vertex planar graphs on \emph{the same vertex set}, $V$, we say that $G_1, \dots, G_p$ have a \emph{$k$-partial simultaneous geometric embedding with mapping} ($k$-PSGE) if there exists a straight-line crossing-free drawing $D_i$ of each $G_i$, $i\in\{1,\dots,p\}$ such that  there exists a subset $V'\subseteq V$ with $|V'|
\geq k$ and each vertex $v\in V'$ is represented by the same point in all drawings $D_i$, $i\in\{1,\dots,p\}$.

If we try to mimic the earlier untangling argument that proves \thmref{psge}, it fails  for $p=3$ already since we cannot guarantee that when $G_3$ is untangled the set of its vertices that stays fixed has a non-empty intersection with the set that remained fixed when untangling $G_2$. It is here that part (d) of \lemref{chain} is needed, or rather the stronger result on column planarity from \thmref{scp}.

\begin{thm}\thmlabel{ppsge}
Any set of $p\geq 2$ $n$-vertex planar graphs has an $\Omega(n^{1/4^{(p-1)}})$-partial simultaneous geometric embedding with mapping, that is, it has an $\Omega(n^{1/4^{(p-1)}})$-PGSE. 
\end{thm}

\begin{proof}
Let $\{G_1, \dots, G_p\}$ be the given set of $p$ $n$-vertex planar graphs. The proof is by induction on $p$. The base case, $p=2$, is true by \thmref{psge}. Let $p\geq 3$ and assume by induction that the set $\{G_1,\dots, G_{p-1}\}$ has an $\Omega(n^{1/4^{(p-2)}})$-PGSE. Let $V'\subseteq V$ be the set from the definition of $k$-PSGE and let $P'$ be the set of $|V'|$ points that $V'$ is mapped to in the drawings $D_1, \dots, D_{p-1}$. Thus $|V'|\in \Omega(n^{1/4^{(p-2)}})$ by induction. We may assume that no pair of points in $P'$ has the same \x-coordinate as otherwise we can just rotate  the union of $D_1, \dots, D_{p-1}$. By \thmref{scp}, there exists $R\subseteq V'$ that is strongly column planar in $G_p$ and $|R|\geq \sqrt{|V'|/3}$.  Since the vertices of $V'$ are bijectively mapped to $P'$, that mapping defines a bijective mapping from $R$ to a subset $P_R$ of $P'$. Consider the total order $\mu$ of $R$ (the total order from the definition of strongly column planar sets) and the total order $\phi$ of $R$ as defined by the \x-coordinates of $P_R$. By the Erd\H{o}s--Szekeres theorem \cite{ES35, steele1995variations}, there exists a subset $R'$ of $R$ of at least $\sqrt{|R|}\geq (|V'|/3)^{1/4}$ vertices such that the order of $R'$ in $\mu$ is the same or reverse as the order of $R'$ in $\phi$. In the second case the union of all the drawings of $D_1, \dots, D_p$ can be mirrored such that the order of $R'$ in $\mu$ is the same as the order of $R'$ in $\phi$. Thus in both cases, we can apply \thmref{scp}. Since the vertices of $R$ are bijectively mapped to $P_R$, this defines a bijective mapping from $R'$ to a subset $P_R'$ of $P_R$. Since $R'$ is strongly column planar in $G_p$, we can apply \thmref{scp} to conclude that $G_p$ has a straight-line crossing-free drawing $D_p$ that respects the mapping from $R'$ to $P_R'$ and thus each vertex $v\in R'$ is represented by the same point in all drawings $D_i$, $i\in\{1,\dots,p\}$. Since $|V'|\in \Omega(n^{1/4^{(p-2)}})$, and $|R'| \geq (|V'|/3)^{1/4}$, the lower bound holds. \qed
\end{proof}

Note that the definition of $k$-PSGE, as introduced in Evans~\etal~\cite{DBLP:conf/gd/EvansKSS14}, has one additional requirement, as compared with the definition used here. Namely, the additional requirement states that if $v, w\in V$ are mapped to a same point in $D_i$ and $D_j$,  then $v=w$. 
However this additional requirement can always be met by the fact that it is possible to perturb any subset of  vertices of a geometric plane graph without introducing crossings. More precisely, for any geometric plane graph there exists a value $\epsilon>0$ such that each vertex can be moved any distance of at most $\epsilon$, and the resulting geometric graph is also crossing-free.\footnote{The maximum value $\epsilon$ for which this property holds is called the tolerance of the arrangement of segments. This concept, both for the geometric realization and the combinatorial meaning of the graphs was systematically studied in \cite{Ramos-phd, ferran-pedro}.} 

\section{Conclusion}


The main purpose of this note is to draw attention to \lemref{chain} and \lemref{antichain} in the current form as they seem to have applications to numerous, some seemingly unrelated, graph drawing problems as evidenced by the results highlighted in the previous sections. The two lemmas appear in the current form for the first time here. Their original formulation was tailored towards specific application (untangling) and not directly applicable to any of the above mentioned problems. 



\paragraph{\bf Acknowledgement.} 
Many thanks to Pat Morin and David R. Wood for very helpful comments on the preliminary version of this article. Similarly, many thanks to the anonymous referees of GD 2015, especially the one who painstakingly corrected my ever random selection from $\{\textup{the, a}, \{\}\}$.

\bibliographystyle{myBibliographyStyle}
\bibliography{paper,2trees}

\end{document}